\documentclass[11pt]{article}
\pdfoutput=1

\usepackage[margin=1in]{geometry}

\DeclareMathAlphabet{\mathbbold}{U}{bbold}{m}{n}
\usepackage{amsmath,amsfonts,amssymb,amsthm}
\usepackage{mathtools}
\usepackage[usenames,dvipsnames,svgnames,table]{xcolor}
\usepackage{thm-restate}

\usepackage[pagebackref]{hyperref}
\hypersetup{
    pdftitle={Randomized query complexity of sabotaged and composed functions}, 
    pdfauthor={Shalev Ben-David and Robin Kothari}, 
    colorlinks=true, 
    linkcolor=blue, 
    citecolor=blue, 
    urlcolor=blue 
}

\renewcommand{\backref}[1]{}

\renewcommand{\backrefalt}[4]{%
\ifcase #1 %
\or
[p.\ #2]%
\else
[pp.\ #2]%
\fi}

\usepackage{tikz,tikz-qtree}
\usepackage{wrapfig}
\usepackage{enumitem}
\usepackage{tocloft}

\makeatletter
\newcommand*\rel@kern[1]{\kern#1\dimexpr\macc@kerna}
\newcommand*\widebar[1]{%
  \begingroup
  \def\mathaccent##1##2{%
    \rel@kern{0.8}%
    \overline{\rel@kern{-0.8}\macc@nucleus\rel@kern{0.2}}%
    \rel@kern{-0.2}%
  }%
  \macc@depth\@ne
  \let\math@bgroup\@empty \let\math@egroup\macc@set@skewchar
  \mathsurround\z@ \frozen@everymath{\mathgroup\macc@group\relax}%
  \macc@set@skewchar\relax
  \let\mathaccentV\macc@nested@a
  \macc@nested@a\relax111{#1}%
  \endgroup
}
\makeatother
\renewcommand{\bar}{\widebar}

\newtheorem{theorem}{Theorem}

\newtheorem{lemma}[theorem]{Lemma}

\newtheorem{corollary}[theorem]{Corollary}
\newtheorem{definition}[theorem]{Definition}

\newenvironment{namedtheorem}[2]
	{\begin{trivlist}\item {\bf #1} (#2){\bf .}\em}{\end{trivlist}}

\theoremstyle{definition}
\newtheorem{open}{Open Problem}

\newcommand{\eq}[1]{\hyperref[eq:#1]{(\ref*{eq:#1})}}
\renewcommand{\sec}[1]{\hyperref[sec:#1]{Section~\ref*{sec:#1}}}
\newcommand{\thm}[1]{\hyperref[thm:#1]{Theorem~\ref*{thm:#1}}}
\newcommand{\lem}[1]{\hyperref[lem:#1]{Lemma~\ref*{lem:#1}}}
\newcommand{\defn}[1]{\hyperref[def:#1]{Definition~\ref*{def:#1}}}
\newcommand{\prop}[1]{\hyperref[prop:#1]{Proposition~\ref*{prop:#1}}}
\newcommand{\cor}[1]{\hyperref[cor:#1]{Corollary~\ref*{cor:#1}}}
\newcommand{\fig}[1]{\hyperref[fig:#1]{Figure~\ref*{fig:#1}}}
\newcommand{\tab}[1]{\hyperref[tab:#1]{Table~\ref*{tab:#1}}}
\newcommand{\alg}[1]{\hyperref[alg:#1]{Algorithm~\ref*{alg:#1}}}
\newcommand{\app}[1]{\hyperref[app:#1]{Appendix~\ref*{app:#1}}}

\newcommand{\be}{\begin{equation}}
\newcommand{\ee}{\end{equation}}

\newcommand{\B}{\{0,1\}}
\newcommand{\Ba}{\{0,1,*\}}
\newcommand{\Bao}{\{0,1,*,\dagger\}}

\newcommand{\AND}{\textsc{And}}
\newcommand{\Ind}{\textsc{Ind}}

\newcommand{\OR}{\textsc{Or}}

\DeclareMathOperator{\adeg}{\widetilde{\deg}}
\DeclareMathOperator{\bs}{bs}
\DeclareMathOperator{\RC}{RC}

\DeclareMathOperator{\RS}{RS}
\DeclareMathOperator{\DS}{DS}
\DeclareMathOperator{\QS}{QS}

\DeclareMathOperator{\R}{R}
\DeclareMathOperator{\D}{D}
\DeclareMathOperator{\Q}{Q}
\DeclareMathOperator{\C}{C}

\DeclareMathOperator{\polylog}{polylog}

\DeclareMathOperator{\Dom}{Dom}
\newcommand{\sab}{\mathrm{sab}}
\newcommand{\usab}{\mathrm{usab}}
\newcommand{\ind}{\mathrm{ind}}
\newcommand{\bin}{\mathrm{bin}}
\DeclareMathOperator{\prt}{prt}

\newcommand{\tOmega}{\widetilde{\Omega}}

\newcommand{\CS}{\mathrm{CS}}

\newcommand{\X}{\mathcal{X}}
\newcommand{\Y}{\mathcal{Y}}
\newcommand{\Z}{\mathcal{Z}}

\begin{document}
\title{\vspace{-1em}Randomized query complexity of sabotaged and composed functions}

\author{
Shalev Ben-David\\
\small Massachusetts Institute of Technology\\
\small \texttt{shalev@mit.edu}
\and
Robin Kothari \\
\small Massachusetts Institute of Technology\\
\small \texttt{rkothari@mit.edu}
}

\date{}
\maketitle

\begin{abstract}
We study the composition question for bounded-error randomized query complexity: Is $\R(f \circ g) = \Omega(\R(f)\R(g))$ for all Boolean functions $f$ and $g$? 
We show that inserting a simple Boolean function $h$, whose query complexity is only $\Theta(\log \R(g))$, in between $f$ and $g$ allows us to prove $\R(f\circ h\circ g) = \Omega(\R(f)\R(h)\R(g))$. 

We prove this using a new lower bound measure for randomized query complexity we call randomized sabotage complexity, $\RS(f)$. Randomized sabotage complexity has several desirable properties, such as a perfect composition theorem, $\RS(f \circ g) \geq \RS(f) \RS(g)$, and  a composition theorem with randomized query complexity, $\R(f \circ g) = \Omega(\R(f)\RS(g))$. It is also a quadratically tight lower bound for total functions and can be quadratically superior to the partition bound, the best known general lower bound for randomized query complexity.

Using this technique we also show implications for lifting theorems in communication complexity. We show that a general lifting theorem for zero-error randomized protocols implies a general lifting theorem for bounded-error protocols.
\end{abstract}

 \setlength{\cftbeforesecskip}{1ex}
 \setlength{\cftaftertoctitleskip}{2ex}
 \setlength{\cftbeforetoctitleskip}{2.5ex}
{\small \tableofcontents}

\clearpage

\section{Introduction}
\label{sec:intro}

\subsection{Composition theorems}

A basic structural question that can be asked in any model of computation is whether there can be resource savings when computing the same function on several independent inputs. 
We say a direct sum theorem holds in a model of computation if solving a problem on $n$ independent inputs requires roughly $n$ times the resources needed to solve one instance of the problem. 
Direct sum theorems hold for deterministic and randomized query complexity \cite{JKS10}, fail for circuit size \cite{Pan12}, and remain open for communication complexity \cite{KRW95,BBCR13,FKNN95}.

More generally, instead of merely outputting the $n$ answers, we could compute another function of these $n$ answers. 
If $f$ is an $n$-bit Boolean function and $g$ is an $m$-bit Boolean function, we define the composed function $f\circ g$ to be an $nm$-bit Boolean function such that $f \circ g (x_1, \ldots, x_n) = f(g(x_1),\ldots, g(x_n))$, where each $x_i$ is an $m$-bit string.
The composition question now asks if there can be significant savings in computing $f \circ g$ compared to simply running the best algorithm for $f$ and using the best algorithm for $g$ to evaluate the input bits needed to compute $f$. 
If we let $f$ be the identity function on $n$ bits that just outputs all its inputs, we recover the direct sum problem.

Composition theorems are harder to prove and are known for only a handful of models, such as deterministic~{\cite{Tal13,Mon14}} and quantum query complexity~{\cite{Rei11,LMR+11,Kim12}}. 
More precisely, let $\D(f)$, $\R(f)$, and $\Q(f)$ denote the deterministic, randomized, and quantum query complexities of $f$. Then for all (possibly partial) Boolean%
\footnote{Composition theorems usually fail for trivial reasons for non-Boolean functions. Hence we restrict our attention to Boolean functions, which have domain $\B^n$ (or a subset of $\B^n$) and range $\B$.}
functions $f$ and $g$, we have
\begin{equation}
\D(f\circ g) = \D(f)\D(g) \qquad \mathrm{and} 
\qquad \Q(f \circ g) = \Theta(\Q(f)\Q(g)).
\end{equation}
In contrast, in the randomized setting we only have the upper bound direction, $\R(f \circ g) = O(\R(f)\R(g)\log \R(f))$. Proving a composition theorem for randomized query complexity remains a major open problem. 
\begin{open}
Does it hold that $\R(f \circ g) = \Omega(\R(f)\R(g))$ for all Boolean functions $f$ and $g$?
\end{open}

In this paper we prove something close to a composition theorem for randomized query complexity. 
While we cannot rule out the possibility of synergistic savings in computing $f \circ g$, we show that a composition theorem does hold if we insert a small gadget in between $f$ and $g$ to obfuscate the output of $g$.
Our gadget is ``small'' in the sense that its randomized (and even deterministic) query complexity is $\Theta(\log \R(g))$.
Specifically we choose the index function, which on an input of size  $k+2^k$ interprets the first $k$ bits as an address into the next $2^k$ bits and outputs the bit stored at that address. 
The index function's query complexity is $k+1$ and we choose $k=\Theta(\log \R(g))$.

\begin{theorem}
\label{thm:comp}
Let $f$ and $g$ be (partial) Boolean functions and let $\Ind$ be the index function with $\R(\Ind) = \Theta(\log \R(g))$. Then $\R(f\circ \Ind \circ g) = \Omega(\R(f)\R(\Ind)\R(g)) = \Omega(\R(f)\R(g)\log \R(g))$.
\end{theorem}

\thm{comp} can be used instead of a true composition theorem in many applications.
For example, recently a composition theorem for randomized query complexity was needed in the special case when $g$ is the $\AND$ function \cite{ABK15}. 
Our composition theorem would suffice for this application, since the separation shown there
only changes by a logarithmic factor
if an index gadget is inserted between $f$ and $g$.

We prove \thm{comp} by introducing a new lower bound technique for randomized query complexity. This is not surprising since the composition theorems for deterministic and quantum query complexities are also proved using powerful lower bound techniques for these models, namely the adversary argument and the negative-weights adversary bound \cite{HLS07} respectively. 

\subsection{Sabotage complexity}
\label{sec:intro_sab}

To describe the new lower bound technique, consider the problem of computing a Boolean function $f$ on an input $x\in\B^n$ in the query model.
In this model we have access to an oracle, which when queried with an index $i\in [n]$ responds with $x_i\in \B$.

Imagine that a hypothetical saboteur damages the oracle and makes some of the input bits unreadable. For these input bits the oracle simply responds with a $*$.
We can now view the oracle as storing a string $p\in\Ba^n$ as opposed to a string $x\in\B^n$.
Although it is not possible to determine the true input $x$ from the oracle string $p$, it may still be possible to compute $f(x)$ if all input strings consistent with $p$ evaluate to the same $f$ value.
On the other hand, it is not possible to compute $f(x)$ if $p$ is consistent with a $0$-input and a $1$-input to $f$. We call such a string $p\in\Ba^n$ a \emph{sabotaged input}.
For example, let $f$ be the $\OR$ function that computes the logical $\OR$ of its bits. 
Then $p=00\!*\!0$ is a sabotaged input since it is consistent with the $0$-input $0000$ and the $1$-input $0010$. 
However, $p=01\!*\!0$ is not a sabotaged input since it is only consistent with $1$-inputs to $f$.

Now consider a new problem in which the input is promised to be sabotaged (with respect to a function $f$) and our job is to find the location of a $*$.
Intuitively, any algorithm that solves the original problem $f$ when run on a sabotaged input must discover at least one $*$, since otherwise it would answer the same on $0$- and $1$-inputs consistent with the sabotaged input.
Thus the problem of finding a $*$ in a sabotaged input is no harder than the problem of computing $f$, and hence naturally yields a lower bound on the complexity of computing $f$. As we show later, this intuition can be formalized in several models of computation.

As it stands the problem of finding a $*$ in a sabotaged input has multiple valid outputs, as the location of any star in the input is a valid output. For convenience we define a decision version of the problem as follows: Imagine there are two saboteurs and one of them has sabotaged our input. 
The first saboteur, Asterix, replaces input bits with an asterisk ($*$) and the second, Obelix, uses an obelisk ($\dagger$).
Promised that the input has been sabotaged exclusively by one of Asterix or Obelix, our job is to identify the saboteur.
This is now a decision problem since there are only two valid outputs.
We call this decision problem $f_\sab$, the \emph{sabotage problem} associated with $f$.

We now define lower bound measures for various models using $f_\sab$. For example, we can define the \emph{deterministic sabotage complexity of $f$} as $\DS(f)\coloneqq \D(f_\sab)$ and in fact, it turns out that for all $f$, $\DS(f)$ equals $\D(f)$ (\thm{DS}).

We could define the \emph{randomized sabotage complexity of $f$} as $\R(f_\sab)$, but instead we define it as $\RS(f)\coloneqq \R_0(f_\sab)$, where $\R_0$ denotes zero-error randomized query complexity, since $\R(f_\sab)$ and $\R_0(f_\sab)$ are equal up to constant factors (\thm{RSR0S}).  
$\RS(f)$ has the following desirable properties.
\begin{enumerate}[itemsep=0.3ex]
    \item \makebox[10em][l]{(Lower bound for $\R$)} For all $f$, $\R(f) = \Omega(\RS(f))$ \hfill (\thm{RgeqRS})
    \item \makebox[10em][l]{(Perfect composition)} For all $f$ and $g$, $\RS(f\circ g) \geq \RS(f) \RS(g)$ \hfill (\thm{RS_compose})
    \item \makebox[10em][l]{(Composition with $\R$)} For all $f$ and $g$, ${\R}(f\circ g) = \Omega(\R(f) \RS(g))$ \hfill (\thm{R_compose})
    \item \makebox[10em][l]{(Superior to $\prt(f)$)} There exists a total $f$ with $\RS(f) \geq \prt(f)^{2-o(1)}$ \hfill (\thm{comparison})
    \item \makebox[10em][l]{(Superior to $\Q(f)$)} There exists a total $f$ with $\RS(f) = \tOmega(\Q(f)^{2.5})$ \hfill (\thm{comparison})
    \item \makebox[10em][l]{(Quadratically tight)} For all total $f$, $\R(f) = O(\RS(f)^2 \log \RS(f))$ \hfill (\thm{rootR0})
\end{enumerate}
Here $\prt(f)$ denotes the partition bound \cite{JK10,JLV14}, which subsumes most other lower bound techniques such as approximate polynomial degree, randomized certificate complexity, block sensitivity, etc. The only general lower bound technique not subsumed by $\prt(f)$ is quantum query complexity, $\Q(f)$, which can also be considerably smaller than $\RS(f)$ for some functions.
In fact, we are unaware of any total function $f$ for which $\RS(f) = o(\R(f))$, leaving open the intriguing possibility that this lower bound technique captures randomized query complexity for total functions.

\subsection{Lifting theorems}

Using randomized sabotage complexity we are also able to show a relationship between lifting theorems in communication complexity.
A lifting theorem relates the query complexity of a function $f$ with the communication complexity of a related function created from $f$.
Recently, G\"o\"os, Pitassi, and Watson \cite{GPW15} showed that there is a communication problem $G_\Ind$, also known as the two-party index gadget, with communication complexity $\Theta(\log n)$ such that for any function $f$ on $n$ bits, $\D^{\mathrm{cc}}(f \circ G_\Ind) = \Omega(\D(f)\log n)$, where $\D^\mathrm{cc}(F)$ denotes the deterministic communication complexity of a communication problem $F$.

Analogous lifting theorems are known for some complexity measures, but no such theorem is known for either zero-error randomized or bounded-error randomized query complexity.
Our second result shows that a lifting theorem for zero-error randomized query complexity implies one for bounded-error randomized query.
We use $\R_0^\mathrm{cc}(F)$ and $\R^\mathrm{cc}(F)$ to denote the zero-error and bounded-error communication complexities of $F$ respectively.

\begin{theorem}
\label{thm:lifting}
Let $G:\X \times \Y \to \B$ be a communication problem with $\min\{|\X|,|\Y|\}= O(\log n)$. 
If it holds that for all $n$-bit partial functions $f$, 
\begin{equation}
\R_0^{\mathrm{cc}}(f \circ G) = \Omega(\R_0(f)/\polylog n),
\end{equation} 
then for all $n$-bit partial functions $f$, 
\begin{equation}
\R^{\mathrm{cc}}(f \circ G_\Ind) = \Omega(\R(f)/\polylog n),
\end{equation}
where $G_\Ind:\B^b \times \B^{2^b} \to \B$ is the index gadget (\defn{commindex}) with $b = \Theta(\log n)$.
\end{theorem}

Proving a lifting theorem for bounded-error randomized query complexity remains an important open problem in communication complexity.
Such a theorem would allow the recent separations in communication complexity shown by Anshu et al.~\cite{ABB+16} to be proved
simply by establishing their query complexity analogues, which
was done in \cite{ABK15} and \cite{AKK15}.
Our result shows that it is sufficient to prove a lifting theorem for zero-error randomized protocols instead.

\subsection{Open problems}

The main open problem is to determine whether $\R(f) = \widetilde{\Theta}(\RS(f))$ for all total functions $f$. 
This is known to be false for partial functions, however. 
Any partial function where all inputs in $\Dom(f)$ are far apart in Hamming distance necessarily has low sabotage complexity. 
For example, any sabotaged input to the collision problem\footnote{In the collision problem, we are given an input $x\in[m]^n$, and we have to decide if $x$ viewed as a function from $[n]$ to $[m]$ is 1-to-1 or 2-to-1 promised that one of these holds.} has at least half the bits sabotaged making $\RS(f) = O(1)$, but $\R(f) = \Omega(\sqrt{n})$.

It would also be interesting to extend the sabotage idea to other models of computation and see if it yields useful lower bound measures. 
For example, we can define quantum sabotage complexity as $\QS(f) \coloneqq  \Q(f_\sab)$, but we were unable to show that it lower bounds $\Q(f)$.

\subsection{Paper organization}

In \sec{prelim}, we present some preliminaries and useful properties of randomized algorithms (whose proofs appear in \app{properties} for completeness). We then formally define sabotage complexity in \sec{sabotage} and prove some basic properties of sabotage complexity. 
In \sec{sum_and_composition} we establish the composition properties of randomized sabotage complexity described above (\thm{RS_compose} and \thm{R_compose}). 
Using these results, we establish the main result (\thm{comp}) in \sec{composition}.
We then prove the connection between lifting theorems (\thm{lifting}) in \sec{lifting}.
In \sec{comparison} we compare randomized sabotage complexity with other lower bound measures.
We end with a discussion of deterministic sabotage complexity in \sec{other}.

\section{Preliminaries}
\label{sec:prelim}

In this section we define some basic notions in query complexity that will be used throughout the paper. 
Note that all the functions in this paper have Boolean input and output, except sabotaged functions whose input alphabet is $\Bao$.
For any positive integer $n$, we define $[n]\coloneqq\{1,2,\ldots,n\}$.

In the model of query complexity, we wish to compute an $n$-bit Boolean function $f$
on an input $x$ given query access to the bits of $x$.
The function $f$ may be total, i.e., $f:\B^n\to\B$, or partial,
which means it is defined only on a subset of $\B^n$, which we denote by $\Dom(f)$.
The goal is to output $f(x)$ using as few queries to the bits of $x$ as possible.
The number of queries used by the best possible deterministic
algorithm (over worst-case choice of $x$) is denoted $\D(f)$.

A randomized algorithm is a probability distribution over deterministic algorithms. 
The worst-case cost of a randomized algorithm is the worst-case (over all the deterministic algorithms in its support) number of queries made by the algorithm on any input $x$.
The expected cost of the algorithm is the expected number of queries
made by the algorithm (over the probability distribution) on an input $x$ maximized over all inputs $x$.
A randomized algorithm has error at most $\epsilon$ if it outputs 
$f(x)$ on every $x$ with probability at least $1-\epsilon$.

We use $\R_\epsilon(f)$ to denote the worst-case cost of the best randomized algorithm that computes $f$ with error $\epsilon$. Similarly, we use $\bar{\R}_\epsilon(f)$ to denote the expected cost of the best randomized algorithm that computes $f$ with error $\epsilon$.
When $\epsilon$ is unspecified it is taken to be $\epsilon=1/3$.
Thus $\R(f)$ denotes the bounded-error randomized query complexity of $f$.
Finally, we also define zero-error expected randomized query complexity, $\bar{\R}_0(f)$,
which we also denote by $\R_0(f)$ to be consistent with the literature.
For precise definitions of these measures as well as the definition of quantum query complexity
$\Q(f)$, see the survey by Buhrman and de Wolf~\cite{BdW02}.

\subsection{Properties of randomized algorithms}

We will assume familiarity with the following basic properties of randomized algorithms. 
For completeness, we prove these properties in \app{properties}.

First, we have Markov's inequality, which allows us to convert an algorithm with a guarantee on the expected number of queries into
an algorithm with a guarantee on the maximum number of queries with a constant factor loss in the query bound and a constant factor increase in the error. This can be used, for example, to convert zero-error randomized algorithms into bounded-error randomized algorithms.

\begin{restatable}[Markov's Inequality]{lemma}{markov}
\label{lem:markov}
Let $A$ be a randomized algorithm that makes $T$ queries in expectation (over its internal randomness). Then for any $\delta\in (0,1)$, the
algorithm $A$ terminates within $\lfloor T/\delta\rfloor$ queries with probability at least $1-\delta$.
\end{restatable}

The next property allows us to amplify the success probability of an $\epsilon$-error randomized algorithm.

\begin{restatable}[Amplification]{lemma}{amplification}
\label{lem:amplification}
If $f$ is a function with Boolean output and $A$ is a randomized
algorithm for $f$ with error $\epsilon<1/2$, repeating $A$ several times and taking the majority vote of the outcomes decreases the error. To reach error $\epsilon^\prime>0$, it suffices to repeat the algorithm
$\frac{2\ln(1/\epsilon^\prime)}{(1-2\epsilon)^2}$ times.
\end{restatable}

Recall that we defined $\bar{\R}_\epsilon(f)$ to be the minimum expected number of queries made by a randomized algorithm
that computes $f$ with error probability at most $\epsilon$.
Clearly, we have $\bar{\R}_\epsilon(f)\leq \R_\epsilon(f)$, since the expected number of queries made by an algorithm is at most the maximum number of queries made by the algorithm.
Using \lem{markov}, we can now relate them in the other direction.

\begin{restatable}{lemma}{Rexpprime}\label{lem:Rexp_prime}
Let $f$ be a partial function, $\delta>0$, and $\epsilon \in [0,1/2)$.
Then we have
$\R_{\epsilon+\delta}(f)\leq \frac{1-2\epsilon}{2\delta}\bar{\R}_{\epsilon}(f) \leq \frac{1}{2\delta}\bar{\R}_{\epsilon}(f)$.
\end{restatable}

The next lemma shows how to relate these measures with the same error $\epsilon$ on both sides of the inequality. 
This also shows that $\bar{\R}_\epsilon(f)$ is only a constant factor away from $\R_\epsilon(f)$ for constant $\epsilon$.

\begin{restatable}{lemma}{Rexp}
\label{lem:Rexp}
If $f$ is a partial function,
then for all $\epsilon\in(0,\frac{1}{2})$, we have
$\R_\epsilon(f)\leq
14\frac{\ln(1/\epsilon)}
{(1-2\epsilon)^2}
\bar{\R}_\epsilon(f)$.
When $\epsilon=\frac{1}{3}$, we can improve this to
$\R(f)\leq 10\bar{\R}(f)$.
\end{restatable}

Although these measures are closely related for constant error, $\bar{\R}_\epsilon(f)$ is more convenient than $\R_\epsilon(f)$ for discussing composition and direct sum theorems.

We can also convert randomized algorithms that
find certificates with bounded error into zero-error
randomized algorithms.

\begin{restatable}{lemma}{repeat}
\label{lem:repeat}
Let $A$ be a randomized algorithm that uses $T$ queries in expectation and finds a certificate with probability $1-\epsilon$. Then repeating $A$ when it fails to find a certificate turns it into an algorithm that always finds a certificate (i.e., a zero-error algorithm) that uses at most $T/(1-\epsilon)$ queries in expectation.
\end{restatable}

Finally, the following lemma is useful for proving lower bounds on randomized algorithms.

\begin{restatable}{lemma}{block}
\label{lem:block}
Let $f$ be a partial function. Let $A$ be a randomized algorithm
that solves $f$ using at most $T$ expected queries and
with error at most $\epsilon$. 
For $x,y\in\Dom(f)$ if $f(x)\neq f(y)$ then
when $A$ is run on $x$,
it must query an entry on which $x$ differs from $y$
with probability at least $1-2\epsilon$.
\end{restatable}

\section{Sabotage complexity}
\label{sec:sabotage}

We now formally define sabotage complexity. Given a (partial or total) $n$-bit Boolean function $f$,
let $P_f\subseteq\{0,1,*\}^n$
be the set of all partial assignments of $f$ that
are consistent with both a $0$-input and a $1$-input. That is, 
for each $p\in P_f$, there exist $x,y\in\Dom(f)$
such that $f(x)\neq f(y)$ and $x_i=y_i=p_i$ whenever
$p_i\neq*$. Let $\smash{P^\dagger_f}\subseteq\{0,1,\dagger\}^n$
be the same as $P_f$,
except using the symbol $\dagger$ instead of $*$.
Observe that $P_f$ and $\smash{P_f^\dagger}$ are disjoint.
Let $Q_f=P_f \cup \smash{P^\dagger_f}\subseteq\{0,1,*,\dagger\}^n$.
We then define $f_\sab$ as follows.

\begin{definition}\label{defn:sabotage}
Let $f$ be an $n$-bit partial function. We define
$f_\sab:Q_f\to\B$ as $f_\sab(q)=0$ if $q\in P_f$
and $f_\sab(q)=1$ if $q\in \smash{P^\dagger_f}$.
\end{definition}

Note that even when $f$ is a total function, $f_\sab$ is always a partial function.
See \sec{intro_sab} for more discussion and motivation
for this definition. Now that we have defined $f_\sab$,
we can define deterministic and randomized sabotage complexity.

\begin{definition}\label{def:RS}
Let $f$ be a partial function.
Then $\DS(f)\coloneqq \D(f_{\sab})$ and $\RS(f)\coloneqq\R_0(f_{\sab})$.
\end{definition}

We will primarily focus on $\RS(f)$ in this work and only discuss $\DS(f)$ in \sec{other}.
To justify defining $\RS(f)$ as $\R_0(f_{\sab})$
instead of $\R(f_{\sab})$ (or $\bar{\R}(f_\sab)$), we now show these definitions
are equivalent up to constant factors.

\begin{theorem}\label{thm:RSR0S}
Let $f$ be a partial function.
Then $\R_0(f_\sab) \geq \bar{\R}_\epsilon(f_{\sab})\geq
(1-2\epsilon)\R_0(f_{\sab})$.
\end{theorem}

\begin{proof}
The first inequality follows trivially. For the second, 
let $x\in Q_f$ be any valid input to $f_{\sab}$.
Let $x'$ be the input $x$ with asterisks replaced with obelisks and vice versa.
Then since $f_\sab(x)\neq f_\sab(x')$, by \lem{block} any $\epsilon$-error
randomized algorithm that solves $f_\sab$ must find a position on which
$x$ and $x'$ differ with probability at least $1-2\epsilon$.
The positions at which they differ are either asterisks or obelisks.
Since $x$ was an arbitrary input, the algorithm must always
find an asterisk or obelisk with probability at least $1-2\epsilon$.
Since finding an asterisk or obelisk is a certificate for $f_\sab$,
by \lem{repeat}, we get a zero-error algorithm for $f_\sab$ that
uses $\bar{\R}_\epsilon(f_{\sab})/(1-2\epsilon)$ expected queries. Thus 
$\R_0(f_{\sab})\leq \bar{\R}_\epsilon(f_{\sab})/(1-2\epsilon)$,
as desired.
\end{proof}

Finally, we prove that $\RS(f)$ is indeed a lower bound on $\R(f)$, i.e., $\R(f) = \Omega(\RS(f))$.

\begin{theorem}\label{thm:RgeqRS}
Let $f$ be a partial function. Then
$\R_\epsilon(f)\geq\bar{\R}_\epsilon(f)\geq(1-2\epsilon)\RS(f)$.
\end{theorem}

\begin{proof}
Let $A$ be a randomized algorithm for $f$
that uses $\bar{\R}_\epsilon(f)$
randomized queries and outputs the correct answer on every input
in $\Dom(f)$ with probability at least $1-\epsilon$.
Now fix a sabotaged input $x$, and let $p$ be the probability
that $A$ finds a $*$ or $\dag$ when run on $x$.
Let $q$ be the probability that $A$ outputs $0$ if it
doesn't find a $*$ or $\dag$ when run on $x$.
Let $x_0$ and $x_1$ be inputs consistent with $x$ such that
$f(x_0)=0$ and $f(x_1)=1$.
Then $A$ outputs $0$ on $x_1$ with probability at least
$q(1-p)$, and $A$ outputs $1$ on $x_0$ with probability
at least $(1-q)(1-p)$. These are both errors, so we have
$q(1-p)\leq\epsilon$ and $(1-q)(1-p)\leq\epsilon$.
Summing them gives
$1-p\leq 2\epsilon$, or $p\geq 1-2\epsilon$.

This means $A$ finds a $*$ entry within
$\bar{\R}_\epsilon(f)$ expected queries with probability
at least $1-2\epsilon$. By \lem{repeat}, we get
$\frac{1}{1-2\epsilon}\bar{\R}_\epsilon(f)\geq \RS(f)$,
or $\bar{\R}_\epsilon(f)\geq(1-2\epsilon)\RS(f)$.
\end{proof}

We also define a variant of $\RS$ where the number
of asterisks (or obelisks) is limited to one. Specifically,
let $U\subseteq\{0,1,*,\dagger\}^n$ be the set of all
partial assignments with exactly one $*$ or $\dagger$.
Formally, $U \coloneqq \{x\in\Bao^n: |\{i\in[n]:x_i \notin \B\}| = 1\}$.

\begin{definition}\label{def:RS1}
Let $f$ be an $n$-bit partial function. We define
$f_\usab$ as the restriction of $f_\sab$ to $U$,
the set of strings with only one asterisk or obelisk.
That is, $f_\usab$ has domain $Q_f\cap U$, but is equal
to $f_\sab$ on its domain. We then define
$\RS_{\mathrm{u}}(f)\coloneqq \R_0(f_\usab)$. If $Q_f\cap U$ is empty,
we define $\RS_{\mathrm{u}}(f)\coloneqq 0$.
\end{definition}

The measure $\RS_{\mathrm{u}}$ will play a key role in our lifting result in \sec{lifting}.
Since $f_\usab$ is a restriction of $f_\sab$ to a promise,
it is clear that its zero-error randomized query complexity
cannot increase, and so $\RS_{\mathrm{u}}(f)\leq\RS(f)$. 
Interestingly, when $f$ is total, $\RS_{\mathrm{u}}(f)$ equals $\RS(f)$.
In other words, when $f$ is total, we may assume without
loss of generality that its sabotaged version has only one
asterisk or obelisk.

\begin{theorem}\label{thm:RS1}
If $f$ is a total function, then $\RS_{\mathrm{u}}(f)=\RS(f)$.
\end{theorem}

\begin{proof}
We already argued that $\RS(f)\geq\RS_{\mathrm{u}}(f)$. To show
$\RS_{\mathrm{u}}(f)\geq\RS(f)$, we argue that any zero-error algorithm
$A$ for $f_\usab$ also solves $f_\sab$.
The main observation we need is that any input to $f_\sab$
can be completed to an input to $f_\usab$ by replacing
some asterisks or obelisks with $0$s and $1$s.
To see this, let $x$ be an input to $f_\sab$. Without
loss of generality, $x\in P_f$. Then there are two strings
$y,z\in\Dom(f)$ that are consistent with $x$,
satisfying $f(y)=0$ and $f(z)=1$.

The strings $y$ and $z$
disagree on some set of bits $B$, and $x$ has a $*$ or $\dag$
 on all of $B$. Consider starting with $y$
and flipping the bits of $B$ one by one, until we reach
the string $z$. At the beginning, we have $f(y)=0$,
and at the end, we reach $f(z)=1$. This means that
at some point in the middle, we must have flipped a bit
that flipped the string from a $0$-input to a $1$-input.
Let $w_0$ and $w_1$ be the inputs where this happens.
They differ in only one bit. If we replace that bit with $*$ or $\dag$,
we get a partial assignment $w$ consistent with both,
so $w\in P_f$. Moreover, $w$ is consistent with $x$.
This means we have completed an arbitrary input to
$f_\sab$ to an input to $f_\usab$, as claimed.

The algorithm $A$, which correctly solves $f_\usab$, when run on $w$ (a valid input to $f_\usab$) must find an asterisk or obelisk in $w$.
Now consider running $A$ on the input $x$ to $f_\sab$ and compare its execution to when it is run on $w$. If $A$ ever queries a position that is different in $x$ and $w$, then it has found an asterisk or obelisk and the algorithm can now halt. If not, then it must find the single asterisk or obelisk present in $w$, which is also present in $x$.
This shows that the slightly modified version of $A$ that halts if it queries an asterisk or obelisk solves $f_\sab$ and hence
$\RS(f)=\R_0(f_\sab)\leq\R_0(f_\usab)=\RS_{\mathrm{u}}(f)$.
\end{proof}

\section{Direct sum and composition theorems}
\label{sec:sum_and_composition}

In this section, we establish the main composition theorems
for randomized sabotage complexity.
To do so, we first need to establish direct sum theorems
for the problem $f_{\sab}$. In fact,
our direct sum theorems hold more generally for 
zero-error randomized query complexity of partial functions 
(and even relations).
To prove this, we will require Yao's minimax theorem \cite{Yao77}.

\begin{theorem}[Minimax]\label{thm:yao}
Let $f$ be an $n$-bit partial function.
There is a distribution $\mu$ over inputs in $\Dom(f)$
such that all zero-error algorithms for $f$ 
use at least $\R_0(f)$ expected queries on $\mu$.
\end{theorem}

We call any distribution $\mu$ that satisfies the assertion in Yao's theorem a \emph{hard distribution} for $f$.

\subsection{Direct sum theorems}
\label{sec:sum}

We start by defining the $m$-fold direct sum of a function $f$, which is simply the function that accepts $m$ inputs to $f$ and outputs $f$ evaluated on all of them.

\begin{definition}\label{def:sum}
Let $f:\Dom(f)\to \Z$, where $\Dom(f)\subseteq \X^n$, be a partial function
with input and output alphabets $\X$ and $\Z$.
The $m$-fold direct sum of $f$ is the partial function $f^{\oplus m}:\Dom(f)^m\to\Z^m$
such that for any $(x_1,x_2,\ldots,x_m) \in \Dom(f)^m$, we have
\be 
f^{\oplus m}(x_1,x_2,\ldots,x_m)=(f(x_1),f(x_2),\ldots,f(x_m)).
\ee 
\end{definition}

We can now prove a direct sum theorem for zero-error randomized and more generally $\epsilon$-error expected randomized 
query complexity, although we only require the result about zero-error algorithms. 
We prove these results for partial functions, but they also hold for arbitrary relations. 

\begin{theorem}[Direct sum]\label{thm:direct_sum}
For any $n$-bit partial function $f$ and any positive
integer $m$, we have
$\R_0(f^{\oplus m})=m\R_0(f)$. Moreover,
if $\mu$ is a hard distribution for $f$
given by \thm{yao}, then $\mu^{\otimes m}$ is
a hard distribution for $f^{\oplus m}$.
Similarly, for $\epsilon$-error randomized algorithms we get $\bar{\R}_\epsilon(f^{\oplus m}) \geq m\bar{\R}_\epsilon(f)$.
\end{theorem}

\begin{proof}
The upper bound follows from running the $\R_0(f)$
algorithm on each of the $m$ inputs to $f$. By linearity
of expectation, this algorithm solves all $m$ inputs
after $m\R_0(f)$ expected queries.

We now prove the lower bound. Let $A$ be a zero-error
randomized algorithm for $f^{\oplus m}$ that uses $T$ expected
queries when run on inputs from $\mu^{\otimes m}$.
We convert $A$ into an algorithm $B$ for $f$
that uses $T/m$ expected queries when run on inputs from $\mu$.

Given an input $x\sim\mu$, the algorithm $B$ generates
$m-1$ additional ``fake'' inputs from $\mu$.
$B$ then shuffles these together with $x$, and runs $A$
on the result. The input to $A$ is then distributed
according to $\mu^{\otimes m}$, so $A$ uses $T$ queries
(in expectation) to solve all $m$ inputs. $B$ then reads the
solution to the true input $x$.

Note that most of the queries $A$ makes are to fake inputs,
so they don't count as real queries. The only real queries $B$
has to make happen when $A$ queries $x$. But since $x$ is
shuffled with the other (indistinguishable) inputs,
the expected number of queries $A$ makes to $x$ is the
same as the expected number of queries $A$ makes to each
fake input; this must equal $T/m$. Thus $B$ makes $T/m$
queries to $x$ (in expectation) before solving it.

Since $B$ is a zero-error randomized algorithm for $f$
that uses $T/m$ expected queries on inputs from $\mu$,
we must have $T/m\geq\R_0(f)$ by \thm{yao}.
Thus $T\geq m\R_0(f)$, as desired.

The same lower bound proof carries through for $\epsilon$-error expected query complexity, $\bar{\R}_\epsilon(f)$, as long as we use a version of Yao's theorem for this model. For completeness, we prove this version of Yao's theorem in \app{yao}.
\end{proof}

\thm{direct_sum} is essentially \cite[Theorem 2]{JKS10}, but our theorem statement looks different since we deal with expected query complexity instead of worst-case query complexity. From \thm{direct_sum}, we can also prove a direct sum theorem for worst-case randomized query complexity since for $\epsilon\in(0,1/2)$,
\begin{equation}
\R_\epsilon(f^{\oplus m}) \geq \bar{\R}_\epsilon(f^{\oplus m}) \geq m \bar{\R}_\epsilon(f) \geq 2\delta m \R_{\epsilon+\delta}(f),
\end{equation}
for any $\delta>0$, where the last inequality used \lem{Rexp_prime}.

For our applications, however,
we will need a strengthened version of this theorem,
which we call a threshold direct sum theorem.

\begin{theorem}[Threshold direct sum]\label{thm:thresh_sum}
Given an input to $f^{\oplus m}$ sampled from $\mu^{\otimes m}$,
we consider solving only some of the $m$ inputs to $f$.
We say an input $x$ to $f$ is solved if a $z$-certificate was queried that proves
$f(x)=z$. Then any randomized algorithm
that takes an expected $T$ queries and solves an expected
$k$ of the $m$ inputs when run on inputs from $\mu^{\otimes m}$
must satisfy $T\geq k\R_0(f)$.
\end{theorem}

\begin{proof}
We prove this by a reduction to \thm{direct_sum}.
Let $A$ be a randomized algorithm that, when run on
an input from $\mu^{\otimes m}$, solves an expected $k$
of the $m$ instances, and halts after
an expected $T$ queries. We note that these expectations
average over both the distribution $\mu^{\otimes m}$
and the internal randomness of $A$.

We now define a randomized algorithm $B$ that solves the $m$-fold direct sum $f^{\oplus m}$ 
with zero error. $B$ works as
follows: given an input to $f^{\oplus m}$, $B$ first runs
 $A$ on that input. Then $B$ checks which of the $m$ instances of $f$ 
were solved by $A$ (by seeing if a certificate proving
the value of $f$ was found for a given instance of $f$). $B$ then runs the
optimal zero-error algorithm for $f$, which makes $\R_0(f)$ expected queries,
on the instances of $f$ that were not solved by $A$.

Let us examine the expected number of queries used by $B$
on an input from $\mu^{\otimes m}$. Recall that a randomized
algorithm is a probability distribution over deterministic
algorithms; we can therefore think of $A$ as a distribution.
For a deterministic algorithm $D\sim A$ and an input
$x$ to $f^{\oplus m}$, we use $D(x)$ to denote the number
of queries used by $D$ on $x$, and $S(D,x)\subseteq [m]$
to denote the set of inputs to $f$ the algorithm $D$ solves
when run on $x$. Then by assumption
\begin{equation}
T=\mathop{\mathbb{E}}_{x\sim\mu^{\otimes m}}
        \mathop{\mathbb{E}}_{D\sim A} D(x) \qquad \textrm{and} \qquad 
k=\mathop{\mathbb{E}}_{x\sim\mu^{\otimes m}}
        \mathop{\mathbb{E}}_{D\sim A} |S(D,x)|.
\end{equation}
Next, let $R$ be the randomized algorithm that uses $\R_0(f)$
expected queries and solves $f$ on any input.
For an input $x$ to $f^{\oplus m}$, we write
$x=x_1x_2\ldots x_m$ with $x_i\in\Dom(f)$. Then the expected number
of queries used by $B$ on input from $\mu^{\otimes m}$
can be written as
\begin{align}
\phantom{=}&  \mathop{\mathbb{E}}_{x\sim\mu^{\otimes m}}
\mathop{\mathbb{E}}_{D\sim A}\left( D(x) +
\mathop{\mathbb{E}}_{D_1\sim R}
\mathop{\mathbb{E}}_{D_2\sim R}
\cdots
\mathop{\mathbb{E}}_{D_m\sim R}
\sum_{i\in [m]\setminus S(D,x)} D_i(x_i)\right)\\
=& \mathop{\mathbb{E}}_{x\sim\mu^{\otimes m}}
\mathop{\mathbb{E}}_{D\sim A}\left( D(x) +
\sum_{i\in[m]\setminus S(D,x)}
    \mathop{\mathbb{E}}_{D_i\sim R} D_i(x_i)\right)\\
\leq &  \mathop{\mathbb{E}}_{x\sim\mu^{\otimes m}}
\mathop{\mathbb{E}}_{D\sim A}\left( D(x) +
\sum_{i\in [m]\setminus S(D,x)} \R_0(f)\right)\\
=&  \mathop{\mathbb{E}}_{x\sim\mu^{\otimes m}}
\mathop{\mathbb{E}}_{D\sim A}\left( D(x) +
(m-|S(D,x)|) \R_0(f)\right)\\
=& \quad   T+(m-k)\R_0(f).
\end{align}

Since $B$ solves the direct sum problem on $\mu^{\otimes m}$, the expected number
of queries it uses is at least $m\R_0(f)$ by \thm{direct_sum}.
Hence $T+(m-k)\R_0(f)\geq m\R_0(f)$, so $T\geq k\R_0(f)$.
\end{proof}

\subsection{Composition theorems}
\label{sec:basic_composition}

Using the direct sum and threshold direct sum theorems we have established, we can now prove composition theorems for randomized sabotage complexity. 
We start with the behavior of $\RS$ itself under composition.

\begin{theorem}\label{thm:RS_compose}
Let $f$ and $g$ be partial functions. Then
$\RS(f\circ g)\geq\RS(f)\RS(g)$.
\end{theorem}

\begin{proof}
Let $A$ be any zero-error algorithm for $(f\circ g)_\sab$, and let $T$
be the expected query complexity of $A$ (maximized over all inputs). 
We turn $A$ into a zero-error algorithm
$B$ for $f_\sab$.

$B$ takes a sabotaged input $x$ for $f$.
It then runs $A$ on a sabotaged input to $f\circ g$
constructed as follows.
Each $0$ bit of $x$ is replaced
with a $0$-input to $g$, each $1$ bit of $x$ is replaced
with a $1$-input to $g$, and each $*$ or $\dag$ of $x$ is replaced
with a sabotaged input to $g$. The sabotaged inputs are
generated from $\mu$, the hard distribution
for $g_\sab$ obtained from \thm{yao}.
The $0$-inputs are generated by
first generating a sabotaged input, and then
selecting a $0$-input consistent with that sabotaged input.
The $1$-inputs are generated analogously.

This is implemented in the following way.
On input $x$, the algorithm $B$ generates $n$ sabotaged inputs
from $\mu$ (the hard distribution for $g_{\sab}$),
where $n$ is the length of the string $x$.
Call these inputs $y_1,y_2,\dots,y_n$.
$B$ then runs the algorithm $A$ on this collection
of $n$ strings, pretending that it is an input to
$f\circ g$, with the following caveat: whenever
$A$ tries to query a $*$ or $\dag$ in an input $y_i$, $B$ instead
queries $x_i$. If $x_i$ is $0$, $B$ selects an input from
$f^{-1}(0)$ consistent with $y_i$, and replaces $y_i$ with
this input. It then returns to $A$ an answer consistent
with the new $y_i$.
If $x_i$ is $1$, $B$ selects a consistent input
from $f^{-1}(1)$ instead. If $x_i$ is a $*$ or $\dag$, $B$ returns a $*$ or $\dag$ respectively.

Now $B$ only makes queries to $x$ when it finds a $*$ or $\dagger$ in an input to $g_\sab$. 
But this solves that instance of $g_\sab$, which was drawn from the hard distribution for $g_\sab$. 
Thus the query complexity of $B$ is upper bounded by the number of instances of $g_\sab$ that can be solved by a $T$-query algorithm with access to $n$ instances of $g_\sab$. 
We know from \thm{thresh_sum} that if $A$ makes $T$ expected queries,
the expected number of $*$ or $\dag$ entries it finds among $y_1,y_2,\dots,y_n$ is at most $T/\RS(g)$.
Hence the expected number of queries $B$ makes to $x$ is at most $T/\RS(g)$.
Thus we have $\RS(f)\leq T/\RS(g)$, which gives
$T\geq \RS(f)\RS(g)$.
\end{proof}

Using this we can lower bound the randomized query complexity of composed functions.
In the following, $f^n$ denotes the function $f$ composed with itself $n$ times, i.e., $f^1 = f$ and $f^{i+1} = f \circ f^{i}$.

\begin{corollary}
Let $f$ be a partial function. Then
$\R(f^n)\geq\RS(f)^n/3$.
\end{corollary}

This follows straightforwardly from observing that $\R(f^n)=\R_{1/3}(f^n)\geq(1-2/3)\RS(f^n)$ (using \thm{RgeqRS}) and $\RS(f^n) \geq \RS(f)^n$ (using \thm{RS_compose}).

We can also prove a composition theorem for zero-error and bounded-error randomized query complexity in terms of randomized sabotage complexity. In particular this yields a composition theorem for $\R(f\circ g)$ when $\R(g)=\Theta(\RS(g))$.

\begin{theorem}\label{thm:R_compose}
Let $f$ and $g$ be partial functions. Then
$\bar{\R}_\epsilon(f\circ g)\geq\bar{\R}_\epsilon(f)\RS(g)$.
\end{theorem}

\begin{proof}
The proof follows a similar argument to the proof of
\thm{RS_compose}. Let $A$ be a randomized
algorithm for $f\circ g$ that uses $T$ expected queries
and makes error $\epsilon$.
We turn $A$ into an algorithm $B$ for $f$
by having $B$ generate inputs from $\mu$, the hard distribution for $g_\sab$,
and feeding them to $A$, as before. The only difference is that this
time, the input $x$ to $B$ is not a sabotaged input.
This means it has no $*$ or $\dag$ entries, so all
the sabotaged inputs that $B$ generates turn into
$0$- or $1$-inputs if $A$ tries to query a $*$ or $\dag$ 
in them.

Since $A$ uses $T$ queries, by \thm{thresh_sum},
it finds at most $T/\RS(g)$ asterisks or obelisks (in expectation).
Therefore, $B$ makes at most $T/\RS(g)$ expected
queries to $x$. Since $B$ is correct whenever $A$ is correct,
its error probability is at most $\epsilon$.
Thus $\bar{\R}_\epsilon(f)\leq T/\RS(g)$, and thus $T\geq \bar{\R}_\epsilon(f)\RS(g)$.
\end{proof}

Setting $\epsilon$ to $0$ yields the following corollary.

\begin{corollary}
Let $f$ and $g$ be partial functions. Then
$\R_0(f\circ g)\geq\R_0(f)\RS(g)$.
\end{corollary}

For the more commonly used $\R(f\circ g)$, we obtain the following composition result.

\begin{corollary}\label{cor:R_compose}
Let $f$ and $g$ be partial functions. Then
$\R(f\circ g)\geq\R(f)\RS(g)/10$.
\end{corollary}

This follows from \lem{Rexp}, which gives $\bar{\R}_{1/3}(f)\geq\R(f)/10$, and \thm{R_compose}, since 
$\R(f\circ g)\geq\bar{\R}_{1/3}(f\circ g) \geq \bar{\R}_{1/3}(f)\RS(g)
\geq \R(f)\RS(g)/10$.

Finally, we can also show an upper bound composition result for randomized sabotage complexity. 

\begin{restatable}{theorem}{uppercompose}
\label{thm:upper_compose}
Let $f$ and $g$ be partial functions. Then
$\RS(f\circ g)\leq \RS(f)\R_0(g)$. We also have
$\RS(f\circ g)=O(\RS(f)\R(g)\log\RS(f))$.
\end{restatable}

\begin{proof}
We describe a simple algorithm for finding a $*$ or $\dag$
in an input to $f\circ g$. Start by running
the optimal algorithm for the sabotage problem of $f$.
This algorithm uses $\RS(f)$ expected queries.
Then whenever this algorithm tries to query a bit,
run the optimal zero-error algorithm for $g$ in the
corresponding input to $g$. 

Now, since the input to $f\circ g$ that we are given
is a sabotaged input, it must be consistent with both
a $0$-input and a $1$-input of $f\circ g$. It follows
that some of the $g$ inputs are sabotaged, and moreover,
if we represent a sabotaged $g$-input by $*$ or $\dag$,
a $0$-input to $g$ by $0$, and a $1$-input to $g$ by $1$,
we get a sabotaged input to $f$. In other words, from the inputs to $g$ we can derive a sabotaged input for $f$.

This means that the outer algorithm runs uses an expected $\RS(f)$
calls to the inner algorithm, and ends up calling the
inner algorithm on a sabotaged input to $g$.
Meanwhile, each call to the inner algorithm
uses an expected $\R_0(g)$ queries,
and will necessarily find a $*$ or $\dag$ if the input it is
run on is sabotaged. Therefore, the described algorithm
will always find a $*$ or $\dag$, and its expected running time
is $\RS(f)\R_0(g)$ by linearity of expectation
and by the independence of the internal randomness
of the two algorithms.

Instead of using a zero-error randomized algorithm for $g$, we can
use a bounded-error randomized algorithm for $g$ as long as its error probability is small.
Since we make $O(\RS(f))$ calls to the inner algorithm, if we boost the bounded-error algorithm's
success probability to make the error much smaller than $1/\RS(f)$ 
(costing an additional $\log \RS(f)$ factor), we will get a bounded-error 
algorithm for $(f\circ g)_\sab$. Since $\R((f\circ g)_\sab)$ is the same as $\RS(f\circ g)$ up to a constant factor (\thm{RSR0S}),

\begin{equation}
\RS(f\circ g)=O(\RS(f)\R(g)\log\RS(f)),
\end{equation}
as desired.
\end{proof}

\section{Composition with the index function}
\label{sec:composition}

We now prove our main result (\thm{comp}) restated more precisely as follows.

\begin{namedtheorem}{Theorem \ref*{thm:comp}}{Precise version}
Let $f$ and $g$ be (partial) functions, and
let $m=\Omega(\R(g)^{1.1})$. Then
$\R(f\circ\Ind_m\circ g)=\Omega(\R(f)\R(g)\log m)
=\Omega(\R(f)\R(\Ind_m)\R(g))$.
\end{namedtheorem}

Before proving this, we formally define the index function.

\begin{definition}[Index function]
The index function on $m$ bits, denoted $\Ind_m:\B^m \to \B$, is defined as follows. Let $c$ be the largest integer such that $c+2^c\leq m$. For any input $x\in \B^m$, let $y$ be the first $c$ bits of $x$ and let $z=z_0z_1\cdots z_{2^c-1}$ be the next $2^c$ bits of $x$. If we interpret $y$ as the binary representation of an integer between $0$ and $2^c-1$, then the output of $\Ind_m(x)$ equals $z_y$.
\end{definition}

To prove \thm{comp}, we also require the strong direct product theorem for randomized query complexity that was established by Drucker \cite{Dru12}.

\begin{theorem}[Strong direct product]\label{thm:dir_prod}
Let $f$ be a partial Boolean function,
and let $k$ be a positive integer.
Then any randomized algorithm for $f^{\oplus k}$ that uses
at most $\gamma^3k\R(f)/11$ queries has success probability
at most $(1/2+\gamma)^k$, for any $\gamma\in(0,1/4)$.
\end{theorem}

The first step to proving $\R(f\circ \Ind \circ g) = \Omega(\R(f)\R(\Ind)\R(g))$ is to establish that $\R(\Ind\circ g)$ is essentially the same as  $\RS(\Ind \circ g)$ if the index gadget is large enough.

\begin{lemma}\label{lem:index}
Let $f$ be a partial Boolean function and let
$m=\Omega(\R(f)^{1.1})$.
Then 
\begin{equation}
\RS(\Ind_m\circ f)=\Omega(\R(f)\log m)
=\Omega(\R(\Ind_m)\R(f)).
\end{equation}
Moreover, if $f^{\oplus c}_{\ind}$
is the defined as the index function on $c+2^c$ bits
composed with $f$ in only the first $c$ bits,
we have $\RS(f^{\oplus c}_{\ind}) \geq \RS_{\mathrm{u}}(f^{\oplus c}_{\ind})=\Omega(c\R(f))$
when $c\geq 1.1\log \R(f)$.
\end{lemma}

Before proving \lem{index}, let us complete the proof of \thm{comp} assuming \lem{index}.

\begin{proof}[Proof of \protect{\thm{comp}}]
By \cor{R_compose}, we have
$\R(f\circ\Ind_m\circ g)\geq \R(f)\RS(\Ind_m\circ g)/10$.
Combining this with \lem{index} gives
$\R(f\circ\Ind_m\circ g)=\Omega(\R(f)\R(g)\log m)$,
as desired.
\end{proof}

We can now complete the argument by proving \lem{index}.

\begin{proof}[Proof of \protect{\lem{index}}]
To understand what the inputs to $(\Ind_m\circ f)_{\sab}$
look like, let us first analyze the function $\Ind_m$. 
We can split an input to $\Ind_m$ into a small index
section and a large array section. To sabotage
an input to $\Ind_m$, it suffices to sabotage
the array element that the index points to (using
only a single $*$ or $\dag$). It follows that to sabotage
an input to $\Ind_m\circ f$, it suffices to sabotage
the input to $f$ at the array element that the index
points to. In other words, we consider sabotaged inputs where 
the only stars in the input are in one array cell whose index is the output
of the first $c$ copies of $f$, where $c$ is the largest integer such that $c+2^c\leq m$. Note that
$c=\log m-\Theta(1)$. 

We now convert any $\RS(\Ind_m\circ f)$ algorithm
into a randomized algorithm for $f^{\oplus c}$.
First, using \lem{markov}, we get a
$2\RS(\Ind_m\circ f)$ query randomized
algorithm that finds a $*$ or $\dag$ with probability $1/2$
if the input is sabotaged. Next, consider running
this algorithm on a non-sabotaged input. It makes
$2\RS(\Ind_m\circ f)$ queries. With probability
$1/2$, one of these queries will be in the array cell
whose index is the true answer to $f^{\oplus c}$ evaluated
on the first $cn$ bits. We can then consider
a new algorithm $A$ that runs the
above algorithm for $2\RS(\Ind_m\circ f)$ queries,
then picks one of the $2\RS(\Ind_m\circ f)$
queries at random, and if that query is in an array
cell, it outputs the index of that cell.
Then $A$ uses $2\RS(\Ind_m\circ f)$ queries
and evaluates $f^{\oplus c}$ with probability
at least $\RS(\Ind_m\circ f)^{-1}/4$.

Next, \thm{dir_prod} implies that for any $\gamma\in(0,1/4)$,
either $A$'s success
probability is smaller than $(1/2+\gamma)^c$,
or else $A$ uses at least
$\gamma^3c\R(f)/11$ queries. This means either
\begin{equation}
\label{eq:dpt}
\RS(\Ind_m\circ f)^{-1}/4 \leq (1/2+\gamma)^{c} 
\qquad \textrm{or} \qquad
2\RS(\Ind_m\circ f)\geq \gamma^3c\R(f)/11.
\end{equation}
Now if we choose $\gamma=0.01$, it is clear that the second inequality in \eq{dpt} yields 
$\RS(\Ind_m\circ f)=\Omega(c \R(f))=\Omega(\R(f)\log m)$ no matter what $m$ (and hence $c$) is chosen to be.

To complete the argument, we show that the first inequality in \eq{dpt} also yields the same. 
Observe that the first inequality is equivalent to
\[
\RS(\Ind_m\circ f)
=\Omega\left(\Bigl(\frac{2}{1+2\gamma}\Bigr)^c\right)
=\Omega\left(\Bigl(\frac{2}{1+2\gamma}\Bigr)^{\log m-\Theta(1)}\right)
=\Omega(m^{\log_2(2/1.02)})
=\Omega(m^{0.97}).
\]
We now have $m^{0.97}=\Omega(m^{0.96}\log m)
=\Omega(\R(f)^{1.1\times 0.96}\log m)
=\Omega(\R(f)\log m)$, as desired.

The lower bound on $\RS_{\mathrm{u}}(f^{\oplus c}_{\ind})$
follows similarly once we makes two observations. First, this argument works
equally well for $f^{\oplus c}_{\ind}$ instead of $\Ind_m \circ f$.
Second, sabotaging the array cell indexed by the outputs 
to the $c$ copies of $f$ in $f^{\oplus c}_{\ind}$ introduces only 
one asterisk or obelisk, so the argument above
lower bounds $\RS_{\mathrm{u}}(f^{\oplus c}_{\ind})$ and not only $\RS(f^{\oplus c}_{\ind})$.
\end{proof}

\section{Relating lifting theorems}
\label{sec:lifting}

In this section we establish \thm{lifting}, which proves that a lifting theorem for zero-error randomized communication complexity implies one for bounded-error randomized communication complexity.

To begin, we introduce the two-party index gadget (also used in \cite{GPW15}).

\begin{definition}[Two-party index gadget]\label{def:commindex}
For any integer $b>0$, and finite set $\Y$, we define the index function $G_\Ind:\B^b \times \Y^{2^b} \to \Y$ as follows. Let  $(x,y)\in\B^b \times \Y^{2^b}$ be an input to $G_\Ind$. Then if we interpret $x$ as the binary representation of an integer between $0$ and $2^b-1$, the function $G_\Ind(x,y)$ evaluates to $y_x$, the $x^\textrm{th}$ letter of y. We also let $G_b$ be the index function with $\Y = \B$ and let $G'_b$ be the index function with $\Y = \Bao$.
\end{definition}

The index gadget is particularly useful in communication complexity because it is ``complete'' for functions with a given value of $\min\{|\X|,|\Y|\}$. 
More precisely, any problem $F:\X \times \Y \to \B$ can be reduced to $G_b$ for $b = \lceil \log \min\{|\X|,|\Y|\}\rceil$. 
To see this, say $|\X|\leq |\Y|$ and let $|\X|=2^b$. 
We now map every input $(x,y)\in\X \times \Y$ to an input $(x',y')$ for $G_b$. 
Since $\X$ has size $2^b$, we can view $x$ as a string in $\B^b$ and set $x' = x$. 
The string $y'=y'_0 y'_1 \cdots y'_{2^b-1} \in \B^{2^b}$ is defined as $y'_x = F(x,y)$.
Hence we can assume without loss of generality that a supposed lifting theorem for zero-error protocols is proved using the two-party index gadget of some size.

Our first step is to lower bound the bounded-error randomized communication complexity
of a function in terms of the zero-error randomized communication complexity of a related function.

\begin{lemma}\label{lem:communication_sab}
Let $f$ be an $n$-bit (partial) Boolean function
and let $G_b:\B^b\times\B^{2^b}\to\B$ be the index gadget with $b=O(\log n)$. Then
\begin{equation}\R^\mathrm{cc}(f\circ G_b)=\Omega\left(
    \frac{\R^\mathrm{cc}_0(f_\usab\circ G^\prime_b)}
    {\log n\log\log n}\right),\end{equation}
where $G^\prime_b$ is the index gadget mapping
$\B^b\times\{0,1,*,\dagger\}^{2^b}$ to $\{0,1,*,\dagger\}$.
\end{lemma}

\begin{proof}
We will use a randomized protocol $A$ for $f\circ G_b$
to construct a zero-error protocol $B$ for
$f_\usab\circ G^\prime_b$. Note the given input to
$f_\usab\circ G^\prime_b$ must have a unique
copy of $G^\prime_b$ that evaluates to $*$ or $\dagger$,
with all other copies evaluating to $0$ or $1$.
The goal of $B$ is to find this copy and determine if
it evaluates to $*$ or $\dagger$. This will evaluate
$f_\usab\circ G^\prime_b$ with zero error.

Note that if we replace all $*$ and $\dagger$ symbols
in Bob's input with $0$ or $1$, we would get a valid input to
to $f\circ G_b$, which we can evaluate using $A$.
Moreover, there is a single special
$*$ or $\dagger$ in Bob's input
that governs the value of this input to $f\circ G_b$ 
no matter how we fix the rest of the $*$ and $\dagger$ symbols.
Without loss of generality, we assume that if the special
symbol is replaced by $0$, the function $f\circ G_b$
evaluates to $0$, and if it is replaced by $1$,
it evaluates to $1$.

We can now binary search to find this special symbol.
There are at most $n 2^b$ asterisks and obelisks in Bob's input.
We can set the left half to $0$ and the right half to $1$,
and evaluate the resulting input using $A$. If the answer
is $0$, the special symbol is on the left half; otherwise,
it is on the right half. We can proceed to binary search
in this way, until we have zoomed in on one gadget that must
contain the special symbol. This requires narrowing
down the search space from $n$ possible gadgets to $1$,
which requires $\log n$ rounds. Each round requires a call
to $A$, times a $O(\log\log n)$ factor for error reduction.
We can therefore find the right gadget with bounded
error, using $O(\R^\mathrm{cc}(f\circ G_b)\log n\log\log n)$
bits of communication.

Once we have found the right gadget, we can certify its validity
by having Alice send the right index to Bob, using $b$ bits
of communication, and Bob can check that it points to an asterisk or obelisk. 
Since we found a certificate with constant
probability, we can use \lem{repeat} to turn this into
a zero-error algorithm. Thus
\begin{equation}\R^\mathrm{cc}_0(f_\usab\circ G^\prime_b)=
O(b+\R^\mathrm{cc}(f\circ G_b)\log n\log\log n).\end{equation}
Since $b=O(\log n)$, we obtain the desired result
$\R^\mathrm{cc}_0(f_\usab\circ G^\prime_b)=
O(\R^\mathrm{cc}(f\circ G_b)\log n\log\log n)$.
\end{proof}

Equipped with this lemma we can prove the connection between lifting theorems (\thm{lifting}), stated more precisely as follows.

\begin{namedtheorem}{Theorem \ref*{thm:lifting}}{Precise version}
Suppose that for all partial Boolean functions $f$
on $n$ bits, we have
\begin{equation}\R_0^\mathrm{cc}(f\circ G_b)=\Omega(\R_0(f)/\polylog n)\end{equation}
with $b=O(\log n)$.
Then for all partial functions Boolean functions, we also have
\begin{equation}\R^\mathrm{cc}(f\circ G_{2b})=\Omega(\R(f)/\polylog n).\end{equation}
The $\polylog n$ loss in the $\R^\mathrm{cc}$ result
is only $\log n\log\log^2 n$ worse than the loss in the
$\R_0^\mathrm{cc}$ hypothesis.
\end{namedtheorem}

\begin{proof}
First we show that
for any function $f$ and positive integer $c$,
\begin{equation}\label{eq:find}
\R^\mathrm{cc}(f\circ G_{2b})=\Omega\left(
    \frac{\R^\mathrm{cc}(f^{\oplus c}_{\ind}\circ G_{2b})}
    {c\log c}\right).
\end{equation}
To see this, note that
we can solve $f^{\oplus c}_{\ind}\circ G_{2b}$
by solving the $c$ copies of $f\circ G_{2b}$
and then examining the appropriate cell of the array.
This uses $c\R^\mathrm{cc}(f\circ G_{2b})$ bits of communication,
times $O(\log c)$ since we must amplify the randomized
protocol to an error of $O(1/c)$.

Next, using \eq{find} and \lem{communication_sab} on
$\R^\mathrm{cc}(f^{\oplus c}_{\ind}\circ G_{2b})$, we get
\begin{equation}\label{eq:Rcc}
\R^\mathrm{cc}(f\circ G_{2b})=\Omega\left(
    \frac{\R^\mathrm{cc}(f^{\oplus c}_{\ind}\circ G_{2b})}
    {c\log c}\right)
=\Omega\left(
    \frac{\R_0^\mathrm{cc}((f^{\oplus c}_{\ind})_\usab
        \circ G^\prime_{2b})}
            {c\log c\log n\log\log n}\right).
\end{equation}
From here we want to use the assumed lifting theorem
for $\R_0$. However, there is a technicality: the gadget
$G^\prime_{2b}$ is not the standard index gadget, and
the function $(f^{\oplus c}_{\ind})_\usab$
does not have Boolean
alphabet. To remedy this, we use two bits to represent
each of the symbols $\{0,1,*,\dagger\}$.
Using this representation, we define a new function
$(f^{\oplus c}_{\ind})_\usab^{\bin}$ on twice as many bits.

We now compare
$(f^{\oplus c}_{\ind})_\usab^{\bin}\circ G_b$
to $(f^{\oplus c}_{\ind})_\usab\circ G^\prime_{2b}$.
Note that the former uses two pointers of size $b$ to index
two bits, while the latter uses one pointer of size $2b$
to index one symbol in $\{0,1,*,\dagger\}$
(which is equivalent to two bits).
It's not hard to see that the former function
is equivalent to the latter function restricted to a promise.
This means the communication complexity of the former
is smaller, and hence
\begin{equation}\label{eq:usab1}
    \R_0^\mathrm{cc}((f^{\oplus c}_{\ind})_\usab
        \circ G^\prime_{2b})
=\Omega({\R_0^\mathrm{cc}((f^{\oplus c}_{\ind})_\usab^{\bin}
        \circ G_b)}).
\end{equation}
We are now ready to use the assumed lifting theorem for $\R_0$.
To be more precise, let's suppose a lifting result
that states $\R_0^\mathrm{cc}(f\circ G_b)=\Omega(\R_0(f)/\ell(n))$
for some function $\ell(n)$. Thus
\begin{equation}\label{eq:usab2}
    {\R_0^\mathrm{cc}((f^{\oplus c}_{\ind})_\usab^{\bin}\circ G_b)}
    =\Omega({\R_0((f^{\oplus c}_{\ind})_\usab^{\bin})}/
            {\ell(n)}).
\end{equation}
We note that
\begin{equation}\label{eq:usab3}
\R_0((f^{\oplus c}_{\ind})_\usab^{\bin})
    =\Omega(\R_0((f^{\oplus c}_{\ind})_\usab))
    =\Omega(\RS_{\mathrm{u}}(f^{\oplus c}_{\ind})).\end{equation}
Setting $c=1.1\log\R(f)$, we have
$\RS_{\mathrm{u}}(f^{\oplus c}_{\ind})=\Omega(c\R(f))$ by \lem{index}. Combining this with \eq{usab1}, \eq{usab2}, and \eq{usab3}, we get
\begin{equation}
    \R_0^\mathrm{cc}((f^{\oplus c}_{\ind})_\usab
        \circ G^\prime_{2b})
=\Omega(c\R(f)/\ell(n)).
\end{equation}
Combining this with \eq{Rcc} yields
\begin{equation}
\R^\mathrm{cc}(f\circ G_{2b})=\Omega\left(
    \frac{c\R(f)}{\ell(n)c\log c\log n\log\log n}\right)
    =\Omega\left(
    \frac{\R(f)}{\ell(n)\log n\log\log^2 n}\right).
\end{equation}
This gives the desired lifting theorem for bounded-error randomized communication with
$\polylog n$ loss that is at most $\log n\log\log^2 n$ worse than the loss in the
assumed $\R_0^\mathrm{cc}$ lifting theorem.
\end{proof}

\section{Comparison with other lower bound methods}
\label{sec:comparison}

\setlength{\intextsep}{0pt}%
\setlength{\columnsep}{18pt}%
\begin{wrapfigure}{r}{0.24\textwidth}
\vspace{-0.25em}
\centering
   \begin{tikzpicture}[x=0.7cm,y=1.25cm]

     \node (R) at(2,2){$\R$};
     \node (RS) at(0,1){$\RS$};
     \node (prt) at(2,1){$\prt$};
     \node (RC) at(1,0){$\RC$};
     \node (Q) at(4,1){$\Q$};
     \node (adeg) at(3,0){$\adeg$};

     \path[-] (RS) edge (R);
     \path[-] (prt) edge  (R);
     \path[-] (Q) edge (R);
     \path[-] (RC) edge (RS);
     \path[-] (RC) edge (prt);
     \path[-] (adeg) edge (Q);
     \path[-] (adeg) edge (prt);
   \end{tikzpicture}
\vspace{-.5em}
    \caption{Lower bounds on $\R(f)$.\label{fig:lower}}
\end{wrapfigure}
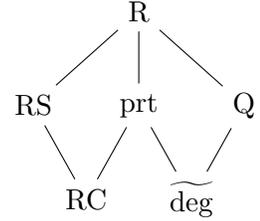

In this section we compare $\RS(f)$ with other lower bound techniques for bounded-error randomized query complexity. 
\fig{lower} shows the two most powerful lower bound techniques for $\R(f)$, the partition bound ($\prt(f)$) and quantum query complexity ($\Q(f)$), which subsume all other general lower bound techniques.
The partition bound and quantum query complexity are incomparable, since there are functions for which the partition bound is larger, e.g., the $\OR$ function, and functions for which quantum query complexity is larger \cite{AKK15}. 
Another common lower bound measure, approximate polynomial degree ($\adeg$) is smaller than both.

Randomized sabotage complexity ($\RS$) can be much larger than the partition bound and quantum query complexity as we now show.
We also show that randomized sabotage complexity is always as large as randomized certificate complexity ($\RC$), which itself is larger than block sensitivity, another common lower bound technique. 
Lastly, we also show that $\R_0(f) = O(\RS(f)^2 \log \RS(f))$, showing that $\RS$ is a quadratically tight lower bound, even for zero-error randomized query complexity.

\subsection{Partition bound and quantum query complexity}

We start by showing the superiority of randomized sabotage complexity against the two best lower bounds for $\R(f)$. 
Informally, what we show is that any separation between $\R(f)$ and a lower bound measure like $\Q(f)$, $\prt(f)$, or $\adeg(f)$ readily gives a similar separation between $\RS(f)$ and the same measure.

\begin{theorem}
\label{thm:comparison}
There exist total functions $f$ and $g$ such that $\RS(f) \geq \prt(f)^{2-o(1)}$ and $\RS(g) = \tOmega(\Q(g)^{2.5})$. There also exists a total function $h$ with $\RS(h) \geq \adeg(h)^{4-o(1)}$.
\end{theorem}

\begin{proof}
These separations were shown with $\R(f)$ in place of
$\RS(f)$ in \cite{ABK15} and \cite{AKK15}.
To get a lower bound on $\RS$, we can simply compose
$\Ind$ with these functions and apply \lem{index}.
This increases $\RS$ to be the same as $\R$ (up to logarithmic
factors), but it does not increase $\prt$, $\adeg$, or $\Q$
more than logarithmically, so the desired separations follow.
\end{proof}

As it turns out, we didn't even need to compose $\Ind$ with
these functions. It suffices to observe that they all use
the cheat sheet construction, and that an argument
similar to the proof of \lem{index}
implies that $\RS(f_{\CS})=\tOmega(\R(f))$ for all $f$
(where $f_{\CS}$ denotes the cheat sheet version of $f$,
as defined in \cite{ABK15}). In particular,
cheat sheets can never be used to separate $\RS$ from $\R$
(by more than logarithmic factors).

\subsection{Randomized certificate complexity}

Finally, we also show that randomized sabotage complexity upper bounds randomized certificate complexity. To show this, we first define randomized certificate complexity.

Given a string $x$, a block is a set of bits of $x$
(that is, a subset of $\{1,2,\ldots,n\}$).
If $B$ is a block and $x$ is a string, we denote by $x^B$
the string given by flipping the bits specified by $B$
in the string $x$. If $x$ and $x^B$ are both in the domain
of a (possibly partial) function $f:\B^n\to\B$
and $f(x) \neq f(x^B)$, we say that $B$
is a sensitive block for $x$ with respect to $f$.

For a string $x$ in the domain $f$, the maximum number
of disjoint sensitive blocks of $x$ is called the
block sensitivity of $x$, denoted by $\bs_x(f)$.
The maximum of $\bs_x(f)$ over all $x$ in the domain
of $f$ is the block sensitivity of $f$, denoted by $\bs(f)$.

A fractionally disjoint set of sensitive blocks of $x$
is an assignment of non-negative weights to the sensitive blocks
of $x$ such that for all $i\in\{1,2,\ldots,n\}$,
the sum of the weights of blocks containing
$i$ is at most $1$. The maximum total weight
of any fractionally disjoint set of sensitive blocks
is called the fractional block sensitivity of $x$.
This is also sometimes called the randomized certificate
complexity of $x$, and is denoted by $\RC_x(f)$ \cite{Aar08,Tal13,GSS16}.
The maximum of this over all $x$ in the domain of $f$
is $\RC(f)$ the randomized certificate complexity of $f$.

Aaronson \cite{Aar08} observed that $\bs_x(f)\leq\RC_x(f)\leq\C_x(f)$. 
We therefore have
\begin{equation}
    \bs(f)\leq\RC(f)\leq\C(f)\leq\R_0(f)\leq\D(f).
\end{equation}
The measure $\RC(f)$ is also a lower bound for $\R(f)$;
indeed, from arguments in \cite{Aar08} it follows that
$\R_\epsilon(f)\geq\RC(f)/(1-2\epsilon)$,
so $\R(f)\geq\RC(f)/3$.

\begin{theorem}\label{thm:RC}
Let $f:\B^n\to\B$ be a partial function. Then
$\RS(f)\geq\RC(f)/4$.
\end{theorem}

\begin{proof}
Let $x$ be the input that maximizes $\RC_x(f)$.
Let $B_1,B_2,\dots B_m$ be all the (not necessarily disjoint)
sensitive blocks of $x$. For each $i \in \{1,2,\dots,m\}$,
let $y_i$ be the sabotaged input formed by replacing
block $B_i$ in $x$ with $*$ entries.
Finding a $*$ in an input chosen from
$Y=\{y_1,y_2,\dots,y_m\}$ is a special case of the sabotage
problem for $f$, so it can be done in $\RS(f)$ expected
queries.

We now use reasoning from \cite{Aar08} to turn this
into a non-adaptive algorithm.
By \lem{markov}, after $\lfloor 2\RS(f)\rfloor$ queries,
we find a $*$ with probability at least $1/2$.
For each $t$ between $1$ and $T=\lfloor2\RS(f)\rfloor$,
let $p_t$ be the probability that the
adaptive algorithm finds a $*$ on query $t$,
conditioned on the previous queries not finding a $*$.
Then we have
\begin{equation}p_1+p_2+\dots+p_T\geq \frac{1}{2}.\end{equation}
If we pick $t\in\{1,2,\dots,T\}$ uniformly and simulate
query $t$ of the adaptive algorithm (which
is possible since we know $x$ and are assuming
the previous $t-1$ queries did not find a $*$),
we must find
a $*$ with probability at least
$1/(2T) \geq 1/(4\RS(f))$.
This is a non-adaptive algorithm for finding a $*$,
so it is also a non-adaptive algorithm for finding
a difference from $x$.

Let the probability distribution over inputs bits obtained from this non-adaptive
algorithm be $(q_1,q_2,\dots,q_n)$, so that the algorithm
queries bit $i$ with probability $q_i$.
We have $\sum_{i=1}^n q_i=1$ and for each sensitive block $B_j$, we have
$\sum_{i\in B_j} q_i\geq 1/(4\RS(f))$.

For each sensitive block $B_j$, let $w_j$ be the weight of $B_j$ under the maximum fractional set of disjoint
blocks. Then $\sum_{j=1}^m w_j=\RC(f)$ and
for each bit $i$, we have $\sum_{j:i\in B_j} w_j\leq 1$.
We then have
\[\frac{\RC(f)}{4\RS(f)}
=\sum_{j=1}^m w_j\cdot\frac{1}{4\RS(f)}
\le \sum_{j=1}^m w_j\sum_{i\in B_j} q_i
=\sum_{i=1}^n q_i\sum_{j:i\in B_j} w_j
\leq \sum_{i=1}^n q_i\cdot 1
\leq 1.\]
Hence $\RS(f)\geq\RC(f)/4$.
\end{proof}

\subsection{Zero-error randomized query complexity}

\begin{theorem}\label{thm:rootR0}
Let $f:\B^n\to\B$ be a total function. Then $\R_0(f) = O(\RS(f)^2 \log \RS(f))$ or alternately, 
$\RS(f)=\Omega(\sqrt{\R_0(f)/\log\R_0(f)})$.
\end{theorem}

\begin{proof}
Let $A$ be the $\RS(f)$ algorithm. The idea
will be to run $A$ on an input to $x$ for long enough
that we can ensure it queries a bit in every sensitive
block of $x$; this will mean $A$ found a certificate
for $x$. That will allow us to turn the algorithm
into a zero-error algorithm for $f$.

Let $x$ be any input, and let $b$ be a sensitive block of $x$.
If we replace the bits of $x$ specified by $b$ with stars,
then we can find a $*$ with probability $1/2$
by running $A$ for $2\RS(f)$ queries
by \lem{markov}. This means that if we run $A$ on $x$
for $2\RS(f)$ queries,
it has at least $1/2$ probability of querying a bit
in any given sensitive block of $x$. If we repeat this $k$ times,
we get a $2k\RS(f)$ query algorithm that queries a bit
in any given sensitive block of $x$ with probability at least
$1-2^{-k}$.

Now, by \cite{KT13}, the number of minimal sensitive blocks in $x$
is at most $\RC(f)^{\bs(f)}$ for a total function $f$.
Our probability of querying a bit in all of these sensitive blocks
is at least $1-2^{-k}\RC(f)^{\bs(f)}$ by the union bound.
When $k\geq 1+\bs(f)\log_2\RC(f)$, this is at least $1/2$.
Since a bit from every sensitive block is a certificate,
by \lem{repeat}, we can turn this into a zero-error
randomized algorithm with expected query complexity
at most $4(1+\bs(f)\log_2\RC(f))\RS(f)$, which gives 
$\R_0(f) = O(\RS(f)\bs(f)\log \RC(f))$. 
Since $\bs(f)\leq\RC(f)=O(\RS(f))$ by \thm{RC}, we have
$\R_0(f)=O(\RS(f)^2\log \RS(f))$,
or
$\RS(f)=\Omega(\sqrt{\R_0(f)/\log\R_0(f)})$.
\end{proof}

\section{Deterministic sabotage complexity}
\label{sec:other}

Finally we look at the deterministic analogue of randomized sabotage complexity. 
It turns out that deterministic sabotage complexity (as defined in \defn{RS}) is exactly the same as deterministic query complexity for all (partial) functions. 
Since we already know perfect composition and direct sum results for deterministic query complexity, it is unclear if deterministic sabotage complexity has any applications.

\begin{theorem}\label{thm:DS}
Let $f:\B^n\to\B$ be a partial function. Then $\DS(f)=\D(f)$.
\end{theorem}

\begin{proof}
For any function $\DS(f) \leq D(f)$ since a deterministic algorithm that correctly computes $f$ must find a $*$ or $\dag$ when run on a sabotaged input, otherwise its output is independent of how the sabotaged bits are filled in. 

To show the other direction, let $\D(f) = k$. This means for every $k-1$ query algorithm, there are two inputs $x$ and $y$ with $f(x) \neq f(y)$, such that they have the same answers to the queries made by the algorithm. If this is not the case then this algorithm computes $f(x)$, contradicting the fact that $\D(f)=k$. Thus if there is a deterministic algorithm for $f_\sab$ that makes $k-1$ queries, there exist two inputs $x$ and $y$ with $f(x) \neq f(y)$ that have the same answers to the queries made by the algorithm. If we fill in the rest of the inputs bits with either asterisks or obelisks, it is clear that this is a sabotaged input (since it can be completed to either $x$ or $y$), but the purported algorithm for $f_\sab$ cannot distinguish them. Hence $\D(f_\sab) \geq k$, which means $\DS(f) \geq \D(f)$.
\end{proof}

\section*{Acknowledgements}
We thank Mika G\"o\"os for finding an error in an earlier proof of \thm{thresh_sum}. We also thank the anonymous referees of Theory of Computing for their comments.

This work is partially supported by ARO grant number W911NF-12-1-0486. This preprint is MIT-CTP \#4806.

\clearpage

\appendix
\section{Properties of randomized algorithms}
\label{app:properties}

We now provide proofs of the properties described in \sec{prelim}, which we restate for convenience.

\markov*

\begin{proof}
If $A$ does not terminate within $\lfloor T/\delta\rfloor$
queries, it must use at least $\lfloor T/\delta\rfloor+1$
queries. Let's say this happens with probability $p$.
Then the expected number of queries used by $A$
is at least $p(\lfloor T/\delta\rfloor+1)$
(using the fact that the number of queries used is always
non-negative). We then get
$T\geq p(\lfloor T/\delta\rfloor+1)>p T/\delta$,
or $p<\delta$. Thus $A$ terminates within $T/\delta$
queries with probability greater than $1-\delta$.
\end{proof}

\amplification*

\begin{proof}
Let's repeat $A$ an odd number of times, say $2k+1$.
The error probability of $A^\prime$, the algorithm
that takes the majority vote of these runs, is
\begin{equation}\sum_{i=0}^k\binom{2k+1}{i}\epsilon^{2k+1-i}(1-\epsilon)^i
\leq \epsilon^{2k+1-k}(1-\epsilon)^k
    \sum_{i=0}^k\binom{2k+1}{i},\end{equation}
which is at most
\begin{equation}\epsilon^{k+1}(1-\epsilon)^k(2^{2k+1}/2)
=\epsilon^{k+1}(1-\epsilon)^k4^k
=\epsilon(4\epsilon(1-\epsilon))^k
=\epsilon(1-(1-2\epsilon)^2)^k.\end{equation}
It suffices to choose $k$ large enough so that
$\epsilon(1-(1-2\epsilon)^2)^k\leq\epsilon^\prime$, or
$\ln\epsilon+k\ln(1-(1-2\epsilon)^2)\leq\ln\epsilon^\prime$.
Using the inequality $\ln(1-x)\leq -x$, it suffices to choose
$k$ so that
$k(1-2\epsilon)^2\geq \ln(1/\epsilon^\prime)-\ln(1/\epsilon)$,
or
\begin{equation}k\geq \frac{\ln(1/\epsilon^\prime)-\ln(1/\epsilon)}
    {(1-2\epsilon)^2}.\end{equation}
In particular, we can choose
\begin{equation}k=\left\lceil\frac{\ln(1/\epsilon^\prime)-\ln(1/\epsilon)}
    {(1-2\epsilon)^2}\right\rceil\leq
    \frac{\ln(1/\epsilon^\prime)}
    {(1-2\epsilon)^2}+1-\frac{\ln(1/\epsilon)}
    {(1-2\epsilon)^2}.\end{equation}
It is not hard to check that
$3(1-2\epsilon)^2\leq 2\ln(1/\epsilon)$ for all
$\epsilon\in(0,1/2)$,
so we can choose $k$ to be at most
$\frac{\ln(1/\epsilon^\prime)}{(1-2\epsilon)^2}-1/2$.
This means $2k+1$ is at most
$\frac{2\ln(1/\epsilon^\prime)}{(1-2\epsilon)^2}$,
as desired.
\end{proof}

\Rexpprime*

\begin{proof}
Let $A$ be the $\bar{\R}_{\epsilon}(f)$ algorithm.
Let $B$ be the algorithm that runs $A$
for $\lfloor \bar{\R}_{\epsilon}(f)/\alpha\rfloor$
queries, and if $A$ doesn't terminate, outputs $0$ with
probability $1/2$ and $1$ with probability $1/2$. Then
by \lem{markov}, the error probability of $B$ is at most
$\alpha/2+(1-\alpha)\epsilon$. If we let
$\alpha = 2\delta/(1-2\epsilon)$, then the error probability of $B$ is at most
\begin{equation}\frac{\delta}{1-2\epsilon}
+\frac{(1-2\epsilon-2\delta)\epsilon}{1-2\epsilon}
=\epsilon+\delta,\end{equation}
as desired. The number of queries made by $B$ is at most $\lfloor \bar{\R}_{\epsilon}(f)/\alpha\rfloor \leq  \frac{1-2\epsilon}{2\delta}\bar{\R}_{\epsilon}(f) \leq \frac{1}{2\delta}\bar{\R}_{\epsilon}(f)$.
\end{proof}

\Rexp*

\begin{proof}
Repeating an algorithm with error $1/3$ three times
decreases its error to $7/27$, so in particular
$\bar{\R}_{7/27}(f)\leq 3\bar{\R}(f)$. Then using
\lem{Rexp_prime} with $\epsilon+\delta=1/3$ and
$\epsilon=7/27$, we get
\begin{equation}\R(f)\leq \frac{1-14/27}{2(1/3-7/27)}
    \bar{\R}_{7/27}(f)
    =\frac{13}{4}\bar{\R}_{7/27}(f)
    \leq\frac{39}{4}\bar{\R}(f)
    \leq 10\bar{\R}(f).\end{equation}

To deal with arbitrary $\epsilon$, we need to use
\lem{amplification}. It gives us
$\R_{\epsilon^\prime}(f)\leq
    \frac{2\ln(1/\epsilon^\prime)}{(1-2\epsilon)^2}
    \R_\epsilon(f)$.
When combined with \lem{Rexp_prime}, this gives
\begin{equation}\R_{\epsilon^\prime}(f)\leq
    \frac{1-2\epsilon^\prime}{\epsilon-\epsilon^\prime}
    \frac{\ln(1/\epsilon^\prime)}{(1-2\epsilon)^2}
    \bar{\R}_{\epsilon^\prime}(f).\end{equation}
Setting $\epsilon=(1+4\epsilon^\prime)/6$
(which is greater than $\epsilon^\prime$
if $\epsilon^\prime < 1/2$) gives
\begin{equation}\R_{\epsilon^\prime}(f)\leq
    \frac{27\ln(1/\epsilon^\prime)}{2(1-2\epsilon^\prime)^2}
    \bar{\R}_{\epsilon^\prime}(f)
\leq 14\frac{\ln(1/\epsilon^\prime)}{(1-2\epsilon^\prime)^2}
    \bar{\R}_{\epsilon^\prime}(f),\end{equation}
which gives the desired result (after exchanging $\epsilon$
and $\epsilon^\prime$).
\end{proof}

\repeat*

\begin{proof}
Let $A^\prime$ be the algorithm that runs $A$, checks if
it found a certificate, and repeats if it didn't.
Let $N_1$ be the random variable for the number of queries
used by $A^\prime$. We know that the maximum number of queries
$A^\prime$ ever uses is the input size; it follows
that $\mathbb{E}(N_1)$ converges and is at most the input size.

Let $M_1$ be the random variable for the number
of queries used by $A$ in the first iteration. Let $S_1$ be the
Bernoulli random variable for the event that $A$ fails to find
a certificate. Then $\mathbb{E}(M_1)=T$ and
$\mathbb{E}(S_1)=\epsilon$.
Let $N_2$ be the random variable for the number
of queries used by $A^\prime$ starting from the second
iteration (conditional on the first iteration failing).
Then
\begin{equation}N_1=M_1+S_1N_2,\end{equation}
so by linearity of expectation and independence,
\begin{equation}\mathbb{E}(N_1)=\mathbb{E}(M_1)+\mathbb{E}(S_1)\mathbb{E}(N_2)
                =T+\epsilon\mathbb{E}(N_2)
                \leq T+\epsilon\mathbb{E}(N_1).\end{equation}
This implies
\begin{equation}\mathbb{E}(N_1)\leq T/(1-\epsilon),\end{equation}
as desired.
\end{proof}

\block*

\begin{proof}
Let $p$ be the probability that $A$ queries an entry
on which $x$ differs from $y$ when it is run on $x$.
Let $q$ be the probability that
$A$ outputs an invalid output for $x$ given that it doesn't
query a difference from $y$. Let $r$ be the probability
that $A$ outputs an invalid output for $y$ given that it
doesn't query such a difference. Since one of these
events always happens,
we have $q+r\geq 1$. Note that $A$ errs with
probability at least $(1-p)q$ when run on $x$ and at least
$(1-p)r$ when run on $y$. This means that
$(1-p)q\leq \epsilon$ and $(1-p)r\leq \epsilon$.
Summing these gives $1-p\leq (1-p)(q+r)\leq 2\epsilon$,
so $p\geq 1-2\epsilon$, as desired.
\end{proof}

\section{Minimax theorem for bounded-error algorithms}
\label{app:yao}

We need the following version of Yao's minimax theorem for
$\bar{\R}_\epsilon(f)$. The proof is similar to other minimax theorems in the literature,
but we include it for completeness.

\begin{theorem}\label{thm:ReYao}
Let $f$ be a partial function and $\epsilon\geq 0$.
Then there exists a distribution $\mu$ over $\Dom(f)$ such that any
randomized algorithm $A$ that computes $f$ with error at most
$\epsilon$ on all $x\in\Dom(f)$ satisfies
$\mathbb{E}_{x\sim \mu} A(x)\geq \bar{\R}_\epsilon(f)$, 
where $A(x)$ is the expected number of queries made by $A$ on $x$.
\end{theorem}

We note that \thm{ReYao} talks only about algorithms that successfully
compute $f$ (with error $\epsilon$) on all inputs, not just
those sampled from $\mu$. An alternative minimax theorem
where the error is with respect to the distribution $\mu$
can be found in \cite{Ver98}, although it loses constant
factors.

\begin{proof}
We think of a randomized algorithm as a probability vector
over deterministic algorithms; thus randomized algorithms lie
in $\mathbb{R}^N$, where $N$ is the number of deterministic
decision trees. In fact, the set $S$ of randomized algorithms
forms a simplex, which is a closed and bounded set.

Let $\mathrm{err}_{f,x}(A) \coloneqq \Pr[A(x)\neq f(x)]$ be
the probability of error of the randomized $A$ when run on $x$.
Then it is not hard to see that
$\mathrm{err}_{f,x}(A)$ is a continuous function of $A$.
Define $\mathrm{err}_f(A)\coloneqq \max_x \mathrm{err}_{f,x}(A)$.
Then $\mathrm{err}_f(A)$ is also the maximum of a finite number
of continuous functions, so it is continuous.

Next consider the set of algorithms
$S_\epsilon\coloneqq\{A\in S:\mathrm{err}_f(A)\leq \epsilon\}$.
Since $\mathrm{err}_f(A)$ is a continuous function and $S$ is closed
and bounded, it follows that $S_\epsilon$ is closed and bounded,
and hence compact. It is also easy to check that $S_\epsilon$
is convex. Let $P$ be the set of probability distributions
over $\Dom(f)$. Then $P$ is also compact and convex.
Finally, consider the function
$\alpha(A,\mu)\coloneqq\mathbb{E}_{x\sim\mu}\mathbb{E}_{D\sim A} D(x)$
that accepts a randomized algorithm and a distribution as input, and returns
the expected number of queries the algorithm makes on that
distribution. It is not hard to see that $\alpha$ is a continuous
function in both variables. In fact, $\alpha$ is linear in both variables 
by the linearity of expectation.

Since $S_\epsilon$ and $P$ are compact and convex subsets of the finite-dimensional spaces $\mathbb{R}^N$ and $\mathbb{R}^{\Dom(f)}$ respectively, and the objective function $\alpha(\cdot,\cdot)$ is linear, we can apply Sion's minimax
theorem (see \cite{Sio58} or \cite[Theorem 1.12]{Wat16}) to get
\begin{equation}
\max_{\mu\in P}\min_{A\in S_\epsilon} \alpha(A,\mu)
=\min_{A\in S_\epsilon}\max_{\mu\in P} \alpha(A,\mu).
\end{equation}
The right hand side is simply the worst-case expected query complexity
of any algorithm computing $f$ with error at most $\epsilon$,
which is $\bar{\R}_\epsilon(f)$ by definition. The
left hand side gives us a distribution $\mu$ such that for
any algorithm $A$ that makes error at most $\epsilon$ on all
$x\in\Dom(f)$, the expected number of queries $A$ makes on $\mu$
is at least $\bar{\R}_\epsilon(f)$.
\end{proof}

\bibliographystyle{alphaurl}
\phantomsection\addcontentsline{toc}{section}{References} 
\newcommand{\eprint}[1]{\small \upshape \tt \href{http://arxiv.org/abs/#1}{#1}}
\let\oldpath\path
\renewcommand{\path}[1]{\small\oldpath{#1}}
\bibliography{sabotage}

\end{document}